\documentclass[acmsmall,nonacm]{acmart}
\usepackage{algorithm}
\usepackage[noend]{algpseudocode}
\usepackage[probability, asymptotics, primitives, notions]{cryptocode}

\usepackage{hyperref}

\usepackage{float}

\algnewcommand\algorithmicswitch{\textbf{switch}}
\algnewcommand\algorithmiccase{\textbf{case}}
\algnewcommand\algorithmicassert{\texttt{assert}}
\algnewcommand\Assert[1]{\State \algorithmicassert(#1)}%
\newcommand\pks{\mathit{pks}}
\newcommand\pkv{\mathit{pkv}}
\newcommand\skv{\mathit{skv}}
\newcommand\tx{\mathit{tx}}

\newcommand\pkb{\mathit{pkb}}
\newcommand\skb{\mathit{skb}}

\newcommand\owners{\mathit{owner}}
\newcommand\validators{\mathit{validators}}

\newcommand\funds{\mathit{funds}}
\newcommand\settled{\mathit{settled}}

\newcommand\tid{\mathit{tx}}
\newcommand\fid{\mathit{id}}
\newcommand\fbl{\mathit{bl}}
\newcommand\fcert{\mathit{certificate}}

\algnewcommand{\LineComment}[1]{\Statex \(\triangleright\) #1}

\newcommand\send{\textbf{send }}

\newcommand\bif{\textbf{if }}

\newcommand\upon{\textbf{upon }}

\newcommand\valid{\textbf{valid}}
\newcommand\invalid{\textbf{invalid}}

\newcommand\rcvd{\textbf{received }}
\newcommand\rcpt{\textbf{receipt of }}

\newlength\myindent
\newcommand\whp{{\em w.h.p. }}

\algdef{SE}[SWITCH]{Switch}{EndSwitch}[1]{\algorithmicswitch\ #1\ }{\algorithmicend\ \algorithmicswitch}%
\algdef{SE}[CASE]{Case}{EndCase}[1]{\algorithmiccase\ #1\textbf{:}}{\algorithmicend\ \algorithmiccase}%
\algtext*{EndSwitch}%
\algtext*{EndCase}%
\algnewcommand{\IfThenElse}[3]{
  \State \algorithmicif\ #1\ \algorithmicthen\ #2\ \algorithmicelse\ #3}
\algdef{SE}[WHEN]{When}{EndWhen}%
   [2]{\algorithmicwhen\ \textwhen{#1}\ifthenelse{\equal{#2}{}}{}{(#2)}}%
   {\algorithmicend\ \algorithmicwhen}%
   \algnewcommand\textwhen{\textnormal}
\algnewcommand\algorithmicwhen{\textbf{when}}   
\algtext*{EndWhen}%

\algdef{SE}[EVERY]{Every}{EndEvery}%
   [2]{\algorithmicevery\ \textevery{#1}\ifthenelse{\equal{#2}{}}{}{(#2)}}%
   {\algorithmicend\ \algorithmicevery}%
\algnewcommand\textevery{\textnormal}
\algnewcommand\algorithmicevery{\textbf{every}}   
\algtext*{EndEvery}%

\newcommand{\mc}[1]{}

\newcommand{\remove}[1]{}

\title{Fractional Payment Transactions: Executing Payment Transactions in Parallel with Less than $f+1$ Validations}

\author{Rida Bazzi}
\orcid{0000-0003-1872-1634}
\affiliation{
  \institution{Arizona State University}
  \city{Tempe}
  \state{AZ}
  \country{USA}
}
\email{bazzi@asu.edu}

\author{Sara Tucci-Piergiovanni}
\orcid{0000-0003-1872-1634}
\affiliation{
  \institution{Universit\'{e} Paris-Saclay, CEA, List}
  \state{Palaiseau}
  \country{France}
  }
\email{sara.tucci@cea.fr}



\begin{document}

\begin{abstract}
We consider the problem of supporting payment transactions in an asynchronous system in which up to $f$ validators are subject to Byzantine failures under the control of an adaptive adversary.  It was shown that, in the case of a single owner, this problem can be solved without consensus by using byzantine quorum systems (requiring a quorum of $2f+1$ validations per transaction). Nonetheless, the process of validating transactions remains sequential. For example, if one has a balance of ten coins and intends to make separate payments of two coins each to two distinct recipients, both transactions must undergo processing by a common correct validator.  On the other hand, these two transactions are non-conflicting as they do not lead to double spending, allowing in principle for parallel validation.
In this paper, we show that it is possible to validate payment transactions in parallel with less than $f$ validations per transaction in an asynchronous system, provided that each transaction spends only a small fraction of a balance. Our solution relies on a novel class of probabilistic quorum systems that we introduce in this paper, termed  \textit{$(k_1,k_2)$-quorum systems}. In the absence of an adaptive adversary, \textit{$(k_1,k_2)$-quorum systems} can be used to enable concurrent and asynchronous validation of up to $k_1$ transactions while preventing validation of more than $k_2$ transactions. In the presence of an adaptive adversary, at least $k_1$ transactions can be validated concurrently, but validation of more than $k'_2 > k_2$  transactions is prevented, with the difference $k'_2-k_2$ dependent on the quorum system's {\em validation slack} -- a term defined in this paper. 
Employing a $(k_1, k_2)$-quorum system, we introduce protocols enabling a payer to validate multiple \textit{fractional spending} transactions in parallel with less than $f+1$ validations per transaction. Subsequently, the payer reclaims any remaining funds through a fully validated transaction, referred to as a \textit{settlement} transaction.
\end{abstract}

\maketitle

\section{Introduction}
Existing cryptocurrencies solve the consensus problem to maintain a shared ledger of all transactions. By having everyone agree on a global order of all transactions in the system, double spending for payment transactions is eliminated and more involved transactions, such as smart contract transactions~\cite{wood2014ethereum}, are made possible. Maintaining a global shared ledger requires a solution to the consensus problem, which is not possible in asynchronous systems~\cite{fischer1985impossibility}. Gupta~\cite{gupta2016} has shown that maintaining a shared ledger is not always strictly needed to support payment transactions. In an asynchronous permissioned system in which at most one third of the validators are subject to Byzantine failures, Byzantine quorums can be used to allow payment transactions -- the asset transfer task. 
The fact that solving consensus is not needed for asset transfer, when the asset has a single owner, was rediscovered by other researchers who formalized the problem, gave it the name {\em asset transfer task} and  generalized it to multi-owner objects~\cite{asset-transfer}.  Other researchers generalized the approach to work in a permissionless proof-of-stake system~\cite{abc-chain}. The key insight of~\cite{gupta2016} is that (1) Byzantine quorums can prevent double spending because any two quorums must have at least one correct validator in their intersection and (2) ordering unrelated transactions is not needed for payment transactions. Consensus-less solutions are important theoretically, but also in practice for their potential to achieve higher throughput and reduced transaction latency~\cite{scalable-reliable-broadcast}. 

Despite the increased efficiency of the consensus-less approach, transaction processing is still fundamentally sequential; If one has a balance of ten coins and wants to pay two coins of the ten coins separately to two different recipients, both transactions would be required to be processed by a common {\em correct} validator, even though the two payments do not lead to double spending. This, in turn, necessitates every request to be processed by a full quorum so that any two quorums have at least $f+1$ validators in common, where $f$ is an upper-bound on the number of faulty validators. The $f+1$-intersection requirement  ensures {\em safety} (no two conflicting payments can be simultaneously validated), but it is enforced even
when payments are non-conflicting. 

In this paper, we introduce and explore the concept of \textit{fractional spending transactions}, which are payment transactions designed to utilize only a small fraction of a balance. These transactions address the common scenario where a balance significantly exceeds individual payments, allowing for separate, non-conflicting payments to different recipients. Our objective is to facilitate the concurrent execution of non-conflicting fractional spending transactions, allowing them to run in parallel. We call the resulting problem, which we define formally in the paper, the {\em fractional spending problem}. This is a challenging problem because: 
(1) we need to allow non-conflicting fractional spending transactions to be validated in parallel without requiring a common correct validator, and (2) we need to disallow such transactions from being validated if there is a potential for double spending. It should not be surprising that it is not possible to spend the whole balance with fractional spending transactions while at the same time satisfying these two
requirements and we provide a simple proof of this impossibility (Lemma \ref{lemma:notfullbalance}). To get around the impossibility, the specification of the fractional spending problem requires the concept of a {\em total spending ratio} which specifies a fraction of a balance that can be spent with fractional spending transactions in parallel without full quorums. For example, if each spending transaction spends $B_{frac}$ fraction of a balance $B$ and the total spending ratio is $s_1\times B_{frac}$, then we can have up to $s_1$ such spending transactions that can be validated in parallel and each of which requiring less than a full quorum. Additional transactions beyond the total spending ratio are not guaranteed to be validated but some of them might still be validated; the {\em buffer} between the total spending ratio and total balance ensures that double spending is not possible by preventing more than $s_2 = B/B_{frac}$ transactions from being validated and ensuring that the total spending not exceed the balance $B$. The parameters $s_1$ and $s_2$ are part of the specification of the fractional spending problem.

A naive solution to allow parallel fractional spending transactions would be to partition the server set into $s_2$ disjoint sets each containing $n/s_2$ validators and to have each fractional spending transaction be validated by one validator set from one partition. The drawback of this approach is that individual partitions have reduced fault tolerance. In fact, if the upper bound on the number of failures is $f$ and $f \geq n/s_2$, an adaptive adversary would have the power to corrupt one whole partition and the solution is not tolerant of $f$ failures. 

Our work focuses on solving the fractional spending problem by introducing parallelism in the validation of non-conflicting transactions, even in the presence of an adaptive adversary in an asynchronous system. Our approach is inherently probabilistic. In fact, we provide a simple proof that a deterministic solution is not possible (Lemma \ref{lemma:nodeterministic}). The validation of fractional spending transactions involves voting quorums drawn from a system of $n$ validators (not partitioned). This system can tolerate up to $f$ faulty validators, where crucially the quorum size may be less than $f+1$ and guaranteed \whp that double spending is not possible. Given that the full balance cannot be spent using fractional payment transactions, any solution to the fractional spending problem should provide clients with the ability to reclaim unspent amounts. Our solution provides a {\em settlement} transaction that can be invoked by the owner at any time to reclaim unspent amounts. The settlement transaction uses larger quorums because otherwise it can lead to double spending. 

At the heart of our solution is a new class of probabilistic quorum systems that we introduce in this paper and that we call $(k_{1},k_{2})$-quorum systems, where $k_{1} < k_{2}$.  Unlike traditional quorum systems (probabilistic or non-probabilistic) in which the main requirement is a lower bound on the size of the intersection of quorums, $(k_{1},k_{2})$-quorum systems have as well an {\em upper-bound} requirement. In $(k_{1},k_{2})$-quorum systems, the upper bound is on the size of the intersection between any quorum and the union of $k_1$ (or less) quorums and the lower bound is on the size of the intersection between any quorum and the union of $k_2$ (or more) quorums.  

Using a $(k_{1},k_{2})$-quorum system, a client would be able to execute with high probability $s_1 = k_1$ fractional spending transactions, each of which spends at most $1/s_2$, $s_2 = k_{2}+f/v_s$ of the client's balance ($v_s$ is the validation slack, which is introduced in the paper), so the total spending ratio is $s_{1}/s_{2}$. As we discussed above, it is possible that the client can succeed in executing more than $s_{1}$ fractional spending transactions but that is not guaranteed. 
Finally, with high probability the client cannot successfully execute more than $s_{2}$ fractional transactions so the total amount cannot exceed the client's balance.  
We propose a construction for $(k_{1},k_{2})$-quorum systems that requires $n > 8f$, which is not optimal.
We believe that exploring the design space for   $(k_{1},k_{2})$-quorum systems is a subject of independent interest, but is beyond the scope of the work. 

We propose a protocol for the fractional spending problem that supports both fractional spending transactions as well as settlement transactions. 
As we explained, settlement transactions allows an owner to reclaim unspent funds and enables the funds owner to use the reclaimed funds in subsequent transactions.
At a high level, interactions between a buyer and a seller (payer and payee) proceeds as follows: 
The buyer has a {\em fund} with an initial balance from which she issues fractional spending transactions. 
The payee contacts a quorum of validators (of size less than $f$) and waits for replies. If enough replies {\em validate} the payment, the seller is guaranteed (with high probability) that it can {\em cash} this payment at any time using a settlement transaction regardless of what other payments the buyer initiates from the fund and independently of whether or not the buyer issues a settlement transaction for the balance of the fund.
To draw a parallel with the banking system, our solution mimics what happens when paying with checks. Our quorum system can concurrently validate multiple checks from the same fund, ensuring the check to be covered, i.e., cashable by the payee. In this system, we consider a fractional payment transaction as \textit{executed} as soon as the payee gathers enough validations. This is because, with the validations, the payee can safely account for the transferred money in its account and can cash the money at any time using a settlement transaction.   Our solution is fully asynchronous and tolerates an adaptive adversary, using a $(k_{1},k_{2})$-quorums for the validation of the fractional payments and a larger quorum for the settlement transaction. 

To summarize the main contributions of this paper are the following:

\begin{enumerate}
    \item We introduce $(k_{1},k_{2})$-quorum systems, a new quorum system that has both a lower bound and an upper bound on intersection requirements. The upper bound allows multiple operations in parallel while the lower bound  limits the number of concurrent operations. The introduction of $(k_{1},k_{2})$-quorum systems is of independent interest and can potentially be applicable to other settings in which we need to allow a limited amount of concurrent activities. In particular, we expect that these systems could be applied to some classes of smart contracts in which available funds significantly exceed actual spending done in small increments.
    \item We introduce the {\em fractional spending problem} which formalizes the requirements on fractional spending transactions and their corresponding settlements.
    \item We present the first protocol that allows multiple payment transactions from any fund to be executed in parallel.
\end{enumerate}


The rest of this paper is organized as follows. Section~\ref{sec:related} discusses related work. Section~\ref{sec:model} presents the system model. Section~\ref{sec:partial-spending} presents the fractional spending problem. Section  ~\ref{sec:k1k2} presents $(k_{1},k_{2})$-Byzantine quorum systems.  
Sections ~\ref{sec:protocol} and~\ref{sec:quorum-select-proofs} present our protocols. Section~\ref{sec:impossible} presents impossibility results, while Section~\ref{sec:conclude} concludes the paper.

\section{Related Work}\label{sec:related}
 The fundamental bottleneck of  blockchains is the underlying  consensus protocol used to add blocks to the replicated data structure.
 To improve blockchain scalability in the context of cryptocurrencies,  two approaches emerged: 
 asynchronous on-chain solutions that attempt to create  consensus-less blockchain protocols \cite{gupta2016,asset-transfer,abc-chain} and off-chain solutions, as channels and channels factories (e.g. \cite{poon2016, Decker15, Pedrosa19,avarikioti2021b}), that move the transaction load offline while 
resorting to a consensus-based blockchain  only for trust establishment and dispute resolution. 

Asynchronous on-chain solutions rely on the assumption that at most $f$ out of $n > 3f$ servers are Byzantine. In such solutions, which only work for single owner funds, payment transactions need to be validated by $2f+1$ validators. It is possible to batch multiple payments to different recipients in one transaction~\cite{DBLP:conf/wdag/NaorK22}, thereby achieving less than $f+1$ validations per payment, but all these payments to different sellers would need to be done at the same time; payments at different times would still need to be validated by $2f+1$ validators each. 
As for off-chain solutions, multiple channels can be used to parallelize payments to different recipients, however, all the parties need to agree, a priori, on the set of participants and  cooperate to use the channels. Cooperation is needed to open the channel (by locking funds as initial balances), update the state of the channel (by signing transactions that attest of the new allocation of balances) and close the channel by sending the last state update to the blockchain and unlock funds. Frauds are avoided  by constantly monitoring the state of the blockchain, in case one other party tries to close the channel with a state update that does not reflect all payments made on the channel. Interestingly, to limit failures and frauds, \cite{avarikioti2021b} uses a set of $n$ processes, called wardens, to proactively validate state updates associated with increasing timestamps agreed by both parties. To get a validation, a full quorum of $t=2f+1$ wardens is needed, under the assumption of at most $f$ out of the $n=3f + 1$ wardens are Byzantine and the non-Byzantine wardens are rational. 
In our solution, we support payments from the same fund to multiple recipients in parallel, 
without prior agreement with the recipients while ensuring (with high probability) no double spending  and requiring less than $f+1$ validations per transaction. In addition, our solution works in the presence of an adaptive adversary that can choose the $f$ validators to corrupt at any time, based on its entire view of the protocol including the entire communication history. 

The closest relative of the $(k_1,k_2)$-quorum systems  that we propose are $k$-quorum systems~\cite{k-quorums,byzantine-k-quorums}. Traditional $k$-quorum systems relax the intersection requirement of quorum systems. A read-quorum is not required to intersect every write-quorum, but they are required to intersect at least one quorum out of any sequence of $k$ successively accessed write-quorums. 
In systems with Byzantine failures, quorums have more than $f$ elements, and the intersection of a read quorum with the $k$ previous write quorums should be large enough to ensure that some correct servers are in the intersection ~\cite{byzantine-k-quorums}. 
The quorum system that we propose in this paper is a probabilistic system in which intersections properties (upper and lower bounds) hold with high probability. Traditional probabilistic quorum systems only have a lower bound requirement~\cite{malkhi1997probabilistic}.
Probabilistic broadcast solutions that can be applied to payment systems have been proposed by others~\cite{DBLP:journals/corr/abs-1908-01738, anikina2023dynamic}, but, to achieve higher performance, they either assume non-adaptive adversaries~\cite{DBLP:journals/corr/abs-1908-01738} or slowly adaptive adversaries and partial synchrony~\cite{anikina2023dynamic}. 

\section{System Model}\label{sec:model}
We consider an asynchronous message passing system of $n$ servers and an unbounded number of clients.  Servers act as {\em validators} for client transactions. Up to $f$ validators and any number of clients can be subject to Byzantine failures under the control of an adaptive adversary. The adversary can corrupt any validators, up to the threshold $f$, as well as any number of  clients. 
If the adversary corrupts a client or a validator, then the adversary has full control of the corrupted party including the contents of its memory before it is corrupted. We refer to the parties under the adversary's control as corrupt or faulty and to the parties not under the adversary's control as honest or correct\footnote{Let us note that the adversary does not have access to the local memory of correct parties.}. Corrupt parties can deviate arbitrarily from their protocols.  %

Message passing is asynchronous. The adversary controls message delay between validators and clients and can delay arbitrarily messages between communicating parties, but cannot delay indefinitely messages between correct parties. In particular, if a correct client or validator sends a request to all validators and waits for responses from $n-f$ validators (as is common in protocols in asynchronous systems), the sender is guaranteed to eventually receive $n-f$ replies. The adversary can select which $n-f$ responses reach the sender first so it would have to act on them without waiting for the remaining responses. 
It follows that if the client's protocol selects a random set $S$ of validators of size $\Omega(n)$, but $S$ is unknown to the adversary, and the client sends a request to all validators and is waiting for $n-f$ responses, with high probability one of the $n-f$ responses that first reach the client must be from an element of $S$. The reason is that the adversary can only guess the identities of some elements of $S$ and selectively delay their communication but cannot do so for all elements of $S$. This is consistent with the goals of our solution whose guarantees hold with high probability. Of course the probability depends on the value of $n$, so the foregoing is understood to hold for large enough $n$.

It is important to note here that we are not changing the standard asynchronous communication model in this paper. Any apparent difference is due to the difference between probabilistic and deterministic solutions. In non-probabilistic protocols, one needs to guarantee correctness in all cases independently of the actions of clients and validators, which are deterministic. In our solution, we allow correct clients to make random choices and keep those choices hidden from the adversary. This is what fundamentally  allows clients to have transactions validated by less than $f+1$ validators. As shown in Lemma \ref{lemma:nodeterministic} deterministic solutions to our problem with less then $f+1$ validations are not possible.

We assume that clients and validators use public-key signature and encryption schemes to sign and encrypt messages and that they are identified by their public keys~\cite{rivest1983method,elgamal1985public}. The public keys of validators are assumed to be known to clients. We assume, but do not explicitly show in the protocols, that messages are signed and signatures are verified. The adversary is computationally bounded and cannot break the encryption or signature schemes. In particular, the adversary cannot read encrypted messages between parties not under its control. Finally,  
parties have access 
to a Hash function in the random oracle model ~\cite{bellare1993random}. 

The adversary can monitor communication between the systems's parties (clients and validators) to determine which parties are communicating together.
We assume that a client (a seller who can have multiple payments from multiple payers) can batch messages for multiple transactions together, so that validators for different transactions are contacted at the same time and the adversary cannot tell, just by observing the communication, which of the contacted validators validate which transactions. This point is discussed further when we present the protocols.

\section{The $(s_1,s_2)$-Fractional Spending Problem}\label{sec:partial-spending}
\subsection{ Funds, Payments and Settlements}
Transactions have two clients, the payer and the payee. We refer to them as the {\em buyer} and the {\em seller}, respectively. We use  $\pkb$ to denote a buyer $b$ and $\pks$ to denote a seller $s$.
Spending is done from funds. A fund $F$ is identified by a unique fund identifier $F.\fid$. It has a balance $F.\fbl$, where $\fbl$ is a non-negative amount of money, and one owner $F.\owners$ associated with it. Funds are {\em certified} by validators. Each fund has a certificate $F.\fcert$ that consists of a set of validations where each validation has the form $(\langle F \rangle,\sigma,\pkv)$ such that $\sigma = \sign_\skv(\langle F \rangle)$, where $\langle F \rangle$ is an encoding of the fund information (id, balance and owners) and $\pkv$ and $\skv$ are the public and private keys of the validator $v$ that signed the validation. 

Depending on the size of $F.\fcert$, we distinguish between {\em fully validated} and {\em partially validated} funds. A fund is fully validated if the certificate has validations from at least $f+1$ validators. A fund is partially validated if it has less than $f+1$ validators.   
Fully validated funds can be initially fully validated through external means, or are the result of settlement transactions (discussed below). We note that our general definition of fully validated funds only ensure that one of the validators is correct, which is necessary but not sufficient to prevent double spending. In our protocols, full validation for settlement transactions requires validation by $n-2f$ validators. 

In general, a fractional spending transaction will be a 
    partially validated transaction that spends a 
    fraction of a balance, which could be specified as part
    of the transaction, up to a limit (recall that spending 
    the whole balance is not possible with partially validated
    transactions).
To keep the presentation simple, we only consider 
    fractional spending transactions that spend a fixed
    fraction of a balance. 
We define a {\em fractional spending transaction} $\tx:(F,\pkb,\pks, \text{PAY})$ as a transfer of money from a fully validated fund $F$ with $\pkb = F.\mathit{owner}$ to a recipient $\pks$. The fund {\em resulting from} a transaction $(F,\pkb,\pks, \text{PAY})$ 
is a partially validated fund $F'$ such that $F'.\fid$ is uniquely determined by $F.\fid$, 
$\pks$, $\pkb$ and a payment number $N_s$ which is uniquely generated by the seller for each payment transaction between $\pks$ and $\pkb$ to allow for multiple payments from the same fund from a buyer to the same seller.
The balance of $F'$ is a fraction of the balance of $F$: $F'.\fbl = F.\fbl/s_2$, for a global constant $s_2$, and $F'.\owners = \{ \pks \}$. 
To keep the presentation simple, we do not support the aggregation of fractional spending transactions from separate funds. 
We note here that the model is different from the UTXO model~\cite{nakamoto2008bitcoin} in which no balance remains in the inputs used for payment. In our model,  fractional spending transactions leave a balance in the fund from which the payments are made, but the balance is not explicitly maintained. 
In addition to payment transactions, there are {\em settlement transactions}. A settlement transaction $(F,\text{SETTLE})$ specifies a fund $F$ to be settled. The fund being settled can be a fully validated or a partially validated fund, but the fund resulting from the execution of a settlement transaction is always fully validated. The identifier of the fund $F'$ resulting from executing  $(F,\text{SETTLE})$ is uniquely determined by the identifier of $F$. 
A validated settlement transaction from fund $F$ results in a fund whose owners are the same as those of $F$ and whose balance is equal to the unspent amount in $F$ and is specified more formally in the problem definition below. 

\subsection{Problem Definition}\label{sec:problem}
 
The formal problem definition below considers both progress and safety requirements for both fractional spending transactions (payment transactions) and settlement transactions. 

For settlements,  funds resulting from payments 
to honest sellers can always be settled successfully. The same is not true for funds resulting from payments to corrupt sellers. In all cases, the total balances of all the funds resulting from settling funds resulting from payments from $F$ does not exceed the balance of $F$. 

We require that the settlement amount for a fund of an honest owner be greater than or equal to the initial balance minus all payments from the fund ; it is not guaranteed to be equal because payments to corrupt sellers could be erased by the adversary.  


Importantly, we have a non-interference requirement stating that if no more than $s \leq s_1$ payments are made from a fund $F$ and no settlement transaction for the fund $F$ is executed, then all $s$ transactions will be executed. This requirement is complemented by a safety requirement, which specifies that at most $s_2$ transactions can be executed from a given fund, ensuring that the total spending does not exceed the fund's balance.

Before formally introducing these  requirements, we introduce some notation. For a given fund $F$, we denote by $\funds_F$ the set of funds resulting from transactions of the form $(F,\pkb,\pks, \text{PAY})$ and we denote by $\funds^H_F$  the set of 
such funds for which the seller $\pks$ is honest. 
We denote by $\settled_F$ the set of fully certified funds resulting from transactions of the form 
$(F_{pay},\text{SETTLE})$, where $F_{pay} \in \funds_F$; i.e. the set of settlement transactions of funds resulting from spending money from $F$. 

The following holds with high probability (we use {\em with high probability}, abbreviated \whp, to denote a probability of the form $1-\negl$, where $\negl$ is a negligible function in the relevant parameters):

    (1) (\textbf{progress}) All partially validated funds with honest owners can be settled successfully: Let $F'$ be a partially validated fund resulting from transaction $(F,\pkb,\pks, \mbox{PAY})$ where $F$ is a fully validated fund and $\pks$ is honest. If $\pks$ executes transaction  $(F',\text{SETTLE})$, the transaction will terminate resulting a fully validated fund. 
    
     (2) (\textbf{safety}) If $F''$ is a fully validated fund resulting from executing $(F',\mbox{SETTLE})$ for partially validated fund $F'$, then $F''.\fbl = F'.\fbl$.  

     (3) (\textbf{safety}) There can be at most $s_2$ fractional spending transactions of the form $(F,\pkb,\pks, \mbox{PAY})$, from a given fund $F$, for a total spending not exceeding $F.\fbl$.  
    
    (4) (\textbf{safety}) Settlement amounts for payments from $F$ are subtracted from settlement for $F$:  If executing transaction $(F,\mbox{SETTLE})$ results in a fully validated fund $F_R$:
   $  F_R.\fbl \leq  F.\fbl - \sum_{F' \in \settled_F} F'.\fbl
   $ 
    
 (5) (\textbf{safety}) Payments to honest sellers are subtracted from the settlement amount: 
    If executing transaction $(F,\mbox{SETTLE})$ results in a fully validated fund $F_R$: $F_R.\fbl \leq  F.\fbl - \sum_{F' \in \funds^H_F} F'.\fbl$
   
 (6) (\textbf{safety}) If the owner of a fund $F$ is honest, the settlement amount is no less than the the initial balance of $F$ minus payments made from $F$:  
    If executing transaction $(F,\mbox{SETTLE})$ results in a fully validated fund $F_R$ and $F.\owners$ is honest:
     $  F_R.\fbl \geq F.\fbl - \sum_{F' \in \funds_F} F'.\fbl
    $

   (7) (\textbf{progress}) Non-interference: If a total of $s \leq s_1$ payment transactions are initiated from fully validated fund $F$ and no additional payment or settlement transactions are initiated by $F.\owners$, then,  every one of the $s$ transactions whose seller (payee) is honest will be validated.  

   (8) (\textbf{progress}) Successful settlement for fully certified funds: If settlement $(F,\mbox{SETTLE})$  of a fully validated fund $F$  is executed by an honest owner, the settlement will terminate.

Note that, for the non-interference requirement, the guarantee of termination for each transaction holds even if the remaining transactions are arbitrarily delayed by slow sellers for example. In other words, executing one transaction does not interfere with the completion of other transactions if the total number of transactions does not exceed the threshold $s_1$.

\section{\texorpdfstring{$(k_{1},k_{2})$}--Byzantine quorums}\label{sec:k1k2}
As we explained, $(k_{1},k_{2})$-Byzantine quorum systems are what enables our solution. 
We first introduce the properties of our new quorum systems, then we propose a construction of a particular quorum system that satisfies these properties. In this section, lemmas and theorems are given without proof. Proofs are presented in the Appendix.

\subsection{Motivation}
Clients send validation requests to a quorum set that acts as a validator set. Some of the validators might have already validated other payments and will not validate the request. Other validators might be corrupt and refuse to validate the request. Finally, due to asynchrony, the client cannot wait for all validators to respond to a given request. A request is considered validated if enough validators validate it. A request is considered not validated if enough validators deny the request. So, we will consider a request to a quorum set $Q$ validated if a fraction $\alpha$ of the validators in $Q$ validate the request. We will consider a request not validated if a fraction $\beta$ of the validators in $Q$ refuse to  validate the request. Of course $\alpha$ and $\beta$ should be chosen in such a way that a request cannot be simultaneously considered to be validated and not validated.

Given that quorums can contain less than $f$ validators, it is important that quorum selection by an honest client be randomized to prevent the adversary from corrupting all validators in the quorum. We assume that quorums are selected according to the uniform access strategy in which all quorums are equally likely to be selected. 

If the adversary is non-adaptive, it selects the $f$ fault servers prior to the commencement of the computation. This means that the number of faulty servers in a given quorum $Q$ will be $f/n \times |Q|$ in expectation and, by Chernoff bounds, the probability that the number of faulty servers deviates significantly from the expectation is negligible in $n$.

If the adversary is adaptive, the adversary's knowledge about the validators depends on whether or not the client that chooses the quorum is honest. 
If the client is honest, then even if the adversary is adaptive, the adversary would not know the identities of the validators in the randomly chosen quorum and the situation is the same as in the case of a non-adaptive adversary. In fact, in the absence of information about the identities of validators, as far as the adversary is concerned, every validator corrupted by the adversary has the same probability of being in the chosen quorum. 
If the client that chooses a quorum is corrupt, then the adversary can learn the identities of the validators and the adversary can selectively corrupt them. This is an issue for the following reason. Assume that the intersection of the chosen quorum with previously chosen quorums is large enough that the transaction would be denied. Since the adversary knows the identities of the validators, the adversary can adaptively corrupt enough validators in the chosen quorum, so that the transaction is validated instead of being denied, effectively flipping the validation decision. The number of validators that need to be corrupted limits the number of transactions with flipped decisions. This is the basis for the definition of the {\em validation slack} which we discuss below.

\subsection{Quorums definition}

\begin{figure}[t]
  \centering
  \includegraphics[width=2.4in]{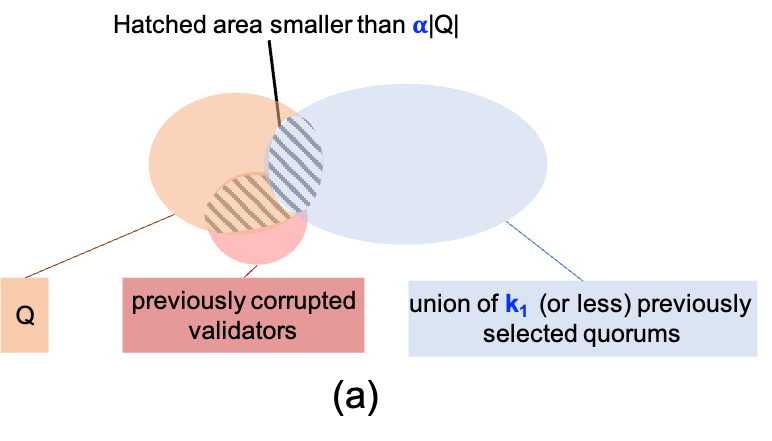} 
    \includegraphics[width=2.7in]{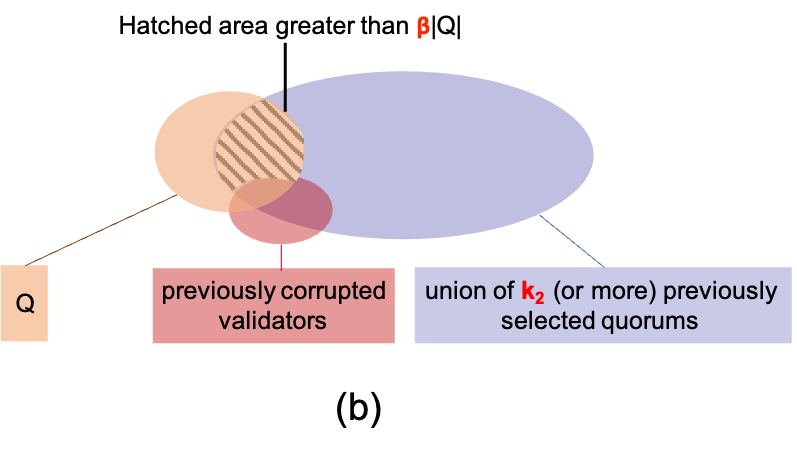}
  \caption{\it Illustration of $(k_1,k2)$-quorums requirements: (a) the size of intersection of $Q$ with up to $k_1$ selected quorums together with previously corrupted servers in $Q$ does not exceed $\alpha |Q|$; (b) the size of the intersection $Q$ with $k_2$ (or more) selected quorums  minus any previously corrupted servers in the intersection is at least $\beta |Q|$;}
  \label{fig:k1k2-illustration}
\end{figure}
We give the definition of  $(k_1,k_2)$-quorum systems as follows.
\begin{definition}
A $(k_1,k_2,\alpha, \beta, \epsilon, \delta)$-quorum system is a collection $\mathcal{Q}$ of sets such that:\\

 (1)     Upper bound  \whp:  $$\Pr_{Q , Q_j \leftarrow  \mathcal{Q}\,;\, j \in J_1; \,|J_1| \leq k_1}[|Q \cap ({\mathcal{F}_{pr}}\cup \bigcup_{j \in J_1} Q_j)| > \alpha |Q| ] \leq \epsilon$$ 
  
  (2)   Lower bound \whp:   $$\Pr_{Q , Q_j \leftarrow \mathcal{Q}\,;\, j \in J_2; \,|J_2| \geq k_2}[|(Q \cap (\bigcup_{j \in J_2} Q_j))-{\mathcal{F}_{pr}}| \leq \beta |Q|]\leq \delta$$ 

where $\mathcal{F}_{pr}$ is the set of faulty servers that existed prior to the random choice of $Q$.
\end{definition}
The lower and upper bound requirements are illustrated in Figure~\ref{fig:k1k2-illustration}. 
The upper and lower bounds ensure that up to $k_1$ accesses create no conflicts and $k_2$ or more accesses create a conflict. 
The definition takes into consideration corrupted servers who can attempt to coordinate their replies to cause the most damage, either to prevent a valid transaction from being validated (by claiming to have handled conflicting transactions) or to allow an invalid transaction to be validated (by validating conflicting transactions). 

The definition contains a number of parameters. The $\epsilon$ and $\delta$ parameters are essentially security parameters that need to be small enough (negligible, in the formal security sense) for the given application. The $\alpha$ and $\beta$ parameters are used for separating the lower and upper bounds. The main parameters from a client point of view are the $k_1$ and $k_2$ parameters, assuming $\epsilon$ and $\delta$ are negligible, hence we refer to these systems as $(k_1,k_2)$-quorum systems.

\noindent\textbf{Adaptive adversary and validation slack.}
If the adversary is adaptive and the client is corrupt, the adversary can ascertain the identities of the validators and manipulate their decisions, causing a request to be validated instead of being denied validation. In order for the adaptive adversary to reverse the decision, it will need to corrupt $(\beta - \alpha)|Q|$ servers. In fact, without the adversary selectively corrupting any servers, $\beta|Q|$ servers consisting of non-corrupt servers will deny the request with high probability.
This is the motivation for defining the {\em validation slack}.

\begin{definition}
The {\em validation slack} $v_s$ of a $(k_1,k_2,\alpha, \beta,\epsilon, \delta)$-quorum system  is $v_s = (\beta - \alpha)|Q_{min}|$, where $Q_{min}$ is the minimum quorum size. 
\end{definition}

In total, in a protocol using a $(k_1,k_2,\alpha, \beta,\epsilon, \delta)$ quorum system, it is possible for at most $k_2 + f/v_s$ accesses to be made:  $k_2$ accesses could go through without the intervention of the adaptive adversary but the additional $f/((\beta - \alpha)|Q_{min}|)$ accesses can only be achieved by adaptively corrupting validators when clients are corrupt.

\subsection{\texorpdfstring{$(k_1,k_2)$}--quorums Construction}
 In this section, we propose a simple construction of a $(k_1,k_2)$-quorum system. 
 Our construction is not optimal and a more careful choice of the system parameters could yield better performance. Exploring the design space for $(k_1,k_2)$-quorum systems is beyond the scope of this work. 

\begin{definition}
An $(m,n)$ {\em uniform balanced} $(k_1,k_2)$-quorum system is a quorum systems in which each quorum has size $m$, $|U| = n = (k_1+k_2)m$ and quorums are selected uniformly at random.
\end{definition}

\newtheorem*{L1}{Lemma~\ref{lem:uniform-prob}}
\begin{lemma}\label{lem:uniform-prob}
An $(m,n)$ uniform balanced quorum system is a $(k_1,k_2)$-quorum system with $\alpha = 1/3$ and $\beta = 2/3$ has $\epsilon = \delta = \negl[n]$ and {\em validation slack} = $m/3$, if $p_f+\alpha_1 < 1/3$, where $\alpha_1 = k_1m/n$ and $p_f = f/n$.
\end{lemma}

\subsection{\texorpdfstring{Asynchronous $(k_1,k_2)$}--quorum systems}

The calculations and definition of $(k_1,k_2)$-quorum systems implicitly assume that all servers in a quorum reply to a request from a client. 
In order to access a quorum in an asynchronous system, the access should allow for non-responding corrupt servers that cannot be timed out. Assuming that the corrupt servers in a given quorum are randomly chosen,
the expected number of corrupt servers in a quorum of size $m$ is at most $m p_f$, where $p_f = f/n$ is an upper bound on the probability that a randomly selected server has been previously corrupted. With high probability, the number of corrupt servers in a quorum is less than $(1+\mu)m p_f$ for any constant $\mu$, $0 < \mu < 1$, and large enough $m$. In other words, the probability that the number of corrupt servers in a quorum exceeds $(1+\mu) m p_f$ is a negligible function of $m$ (and therefore of $n$). It follows that, when contacting a random quorum that is not known to the adversary, an honest client can wait for replies from $m(1-(1+\mu) m p_f)$ servers and be guaranteed \whp to receive that many replies. 
Of those replies, $m(1-2(1+\mu) m p_f)$ are guaranteed \whp to be from non-corrupt servers. 
These considerations affect our definition for $(k_1,k_2)$-quorums in asynchronous systems. The intersection and non-intersection properties should allow for the exclusion of $(1+\mu) m p_f$ servers, which could be correct.  We revise the definition as follows.

\begin{definition}
A $(k_1,k_2,\alpha, \beta, \epsilon, \delta, \mu)$- {\em asynchronous} quorum system is a collection of $\mathcal{Q}$ of sets such that $\forall Q_s \,: |Q_s| \leq (1+\mu) p_f m$, the following holds. \\

   (1) Upper bound:  $$ \Pr_{Q , Q_j \leftarrow  \mathcal{Q}\,;\, j \in J_1; \,|J_1| \leq k_1}[|(Q-Q_s) \cap ({\mathcal{F}_{pr}}\cup \bigcup_{j \in J_1} Q_j)| > \alpha |Q| ] \leq \epsilon$$

    (2) Lower  bound: $$
    \Pr_{Q , Q_j \leftarrow  \mathcal{Q}\,;\, j \in J_2; \,|J_2| \geq k_2}[|((Q-Q_s) \cap (\bigcup_{j \in J_2} Q_j))-{\mathcal{F}_{pr}}| \leq \beta |Q| ] \leq \delta$$

where $\mathcal{F}_{pr}$ is the set of faulty servers that existed prior to the random choice of $Q$.
\end{definition}

This definition requires the intersection properties to hold even if we exclude any set $Q_s$ of size up to $(1+\mu) m p_f$ servers from the intersection. 
We want $\mu$ and $\alpha$ and $\beta$ to be specific constants for which  $\epsilon$ and $\delta$ are small enough. 
Accordingly with this revised definition, we provide the following construction for the asynchronous case.

\newtheorem*{L2}{Lemma~\ref{lem:uniform-prob-asynch}}
\begin{lemma}\label{lem:uniform-prob-asynch}
An $(m,n)$ uniform balanced $(k_1,k_2)$-quorum system is a $(k_1,k_2, \epsilon, \delta, \alpha = 1/3, \beta = 2/3, \mu > 0 )$-asynchronous quorum system with $\epsilon = \delta = \negl[m]$ and {\em validation slack} = $m/3$, if $n > 8f$ and $k_1m/n < 1/24$.
\end{lemma}

\section{Fractional Spending and Settlement Protocols}~\label{sec:protocol} 

\subsection{Protocol Overview}
We assume that we have a fully validated fund to be spent with fractional spending transactions. 
To make a payment, a validation request is sent to all the validators in the quorum and when enough replies are received, the payment is considered validated (partially). To settle a transaction we require full validation. A settlement request for a fully validated fund (from which payments could have been made) is sent to all validators. Every validator broadcasts to all other validators information about payments that it is aware of and then this information is collected by all validators who would then deduct these payments from the balance of the fund resulting from the settlement. To settle a partially validated fund resulting from a payment transaction, validators need to ensure that the transaction was properly validated. 
Many details are needed to make this work. 
We start with some solution considerations.

{\bf Seller should select the quorum:} 
The reason why the seller should be the one selecting the quorum has to do with the fact that their interest is in ensuring that the transaction can be successfully settled whereas a (corrupt) buyer's interest is in erasing records of payments, which can be done if the adversary learns the identities of the validators in the quorum, and the buyer would then be able to settle without the payment counting against the fund's balance. 

{\bf Protecting the seller during settlement:} 
During settlement, the seller needs to fully validate the settlement transaction and therefore needs to divulge the identities of the validators that validated the transaction. 
The seller cannot send the information to all validators because it is possible for one corrupt validator to receive the information before everyone else. At that point, the adversary learns the identities of the validators and corrupt them to erase their record of the transaction. 
A similar scenario can occur during the settlement of a fully validated fund $F$ when validators communicate with each other to determine if there are any fractional spending transactions from $F$ that need to be subtracted from the balance owed. 
To prevent this from happening, divulging the identities of validators is done using  secret sharing~\cite{beimel2011secret} in two stages so that when  the adversary learns the information, it is too late to erase it. 

{\bf Limiting the number of validators:} 
The seller should not get validations from many validators because that prevents other payments from the same fund. To limit the number of validations, we have the seller commit to the identities of the validators (using a hash function) and the buyer signs these commitments that are limited in number. The seller presents the signed commitments to the validators when it validates the transaction. 

 In what follows we assume, for simplicity, that an asynchronous $(k_1,k_2)$-quorum system construction is known to all participants with the parameters $k_1$, $k_2$, $\alpha$, $\beta$, $\mu$ and $m$ available as global constant values. The rest of this section presents an asynchronous solution to $(s_1,s_2)$-fractional spending in the presence of an adaptive adversary (where $s_1 = k_1$ and $s_2 = k_2 + f/v_s$).

\begin{algorithm}[t]
    \caption{Fractional spending: Buyer's Protocol}
    \label{buyer-partial-payment}
    \begin{algorithmic}[1]    
        \Procedure{BuyerFractionalSpend}{$\pks$} 
          
            \State $\tid :=  \langle F, \pkb ,\pks \rangle$ \color{blue}\Comment{transaction id} \color{black}
            \State \textbf{send} $(\mbox{PAY},\tid)$ \textbf{to} $\pks$ \color{blue}\Comment{ payment message with tx Id} \color{black}
           
            \State \textbf{wait} \color{blue} \Comment{wait until commitments received}\color{black}
            \State \textbf{until} $(\mbox{QUORUM},  \tid,h_s, Q_c = [c_i])$ \rcvd \textbf{from} $\pks$

                \If{$|Q_c| = m$} \color{blue}\Comment{if quorum has the correct size,}\color{black}

                    \State  \textbf{send} $(\mbox{SIGNED\_Q},\tid,[\sign_\skb(\tid||h_s||c_i)]$) to $\pks$  

                \Else \color{blue}\Comment{sign and send response, }\color{black}
                    \State \textbf{return $\bot$} \color{blue}\Comment{otherwise abort}\color{black}

                \EndIf
        \EndProcedure
    \end{algorithmic}
\end{algorithm}

\begin{algorithm}[t]
    \caption{Fractional spending: Seller's Protocol}
    \label{seller-partial-payment}
    \begin{algorithmic}[1]  
        \Procedure{SellerFractionalSpending}{$\tid$}

\item[]
                \color{blue}\LineComment{seller selects $m$ validators}\color{black}
            \State $[v_i] = Q_{\tx} \gets \Call{SelectQuorum}{\tid, N_s \leftarrow \{0,1\}^r}$

        \item[]
            \State $[N_i] \leftarrow [(\{0,1\}^r)^m]$  \color{blue}\Comment{blinding nonce for i'th validator}\color{black}\label{Ns}
            \State $[c_i] \gets \Call{H}{v_i||N_i}$ \color{blue}\Comment{ commitment to i'th validator identity}\color{black}
            \State $h_s = \Call{H}{N_s}$ \color{blue}\Comment{commitment to $N_s$}\color{black}

  \item[]
         \color{blue}\LineComment{send commitments to buyer and wait for signatures, then}
         \LineComment{send signatures with nonces to validators}\color{black}
            \State $\send (\mbox{QUORUM},\tid,h_s,[c_i]) \textbf{ to } \pkb$ 
        \State \textbf{wait until} $(\mbox{SIGNED\_Q},\tid,[\sigma_i])$ \rcvd \textbf{from} $\pkb$         
        \ForAll{$v \in Q_{tx}$} 
            \State $\send \langle \tid,h_s,\sigma_i,N_i \rangle \textbf{ to } v$                    
        \EndFor

          \item[]
    \color{blue}\LineComment{wait for replies from validators while keeping track of}
    \LineComment{which validators replied and which are witnesses}\color{black}      
		\State $\mathit{replies} = \mathit{witnesses} = \emptyset$  
        \Repeat 
            \If{$\rcvd \mathit{resp} \textbf{ from } v \in Q_{tx} \wedge v \not\in \mathit{replies}$} 
                \State $\mathit{replies} = \mathit{replies} \cup \{\mathit{v}\} $ 
            \If{$\mathit{resp}  = \langle \valid, \tid,h_s, \sigma = \sign_\skv(\tid||h_s)\rangle$} 
                \State $\mathit{witnesses} = \mathit{witnesses} \cup (\mathit{resp},v) $ 
            \EndIf
            \color{black}
            \EndIf
        \Until{$|\mathit{replies}| \geq  m - (1+\mu) p_f m$} 
        
\item[]
        \color{blue}\LineComment{If enough validators validated, return certificate for}
        \LineComment{fractional spending transaction}\color{black}
        \If{$|\mathit{witnesses}| \geq (1 - \alpha) \times m$}         
            \State \textbf{return } $\langle \tid, N_s,\mathit{witnesses}\rangle $
        \Else
            \State \textbf{return $\bot$}
        \EndIf
        \EndProcedure
    \end{algorithmic}
\end{algorithm}

\subsection{Fractional Spending Protocols}
The fractional spending protocols for the buyer and seller are shown in Algorithms~\ref{buyer-partial-payment} and~\ref{seller-partial-payment} respectively. The validator code relating to fractional spending is shown in Algorithm~\ref{validator-payment}. 
The code closely follows the solution overview above. 
In the code we assume that the quorum parameters including the validation slack $v_s$ are fixed and known by all parties and therefore the fractional amount to be paid from fund $F$, which is equal to  $F.\fbl/s_2 = F.\fbl/(k_2+f/v_s)$, needs not be explicitly shown in the code.

\subsubsection{Buyer's Code} 
The buyer sends a PAY message specifying the transaction identifier  
and waits for a response from the seller
which will be a vector of commitments $[c_i]$ to the identities of a quorum of validators. 
The buyer checks that the quorum size is $m$ and sends a reply that includes for each validator  a signature of $\tid||c_i$ to link the commitment to the transaction.
\begin{algorithm}[t]
    \caption{Validator $v$ Payment Validation}
    \label{validator-payment}
    \begin{algorithmic}[1] 
     \If{ $\,\,\,\,\,\,\,\,\,\,\,$\textbf{received} $\langle \tid,h_s,\sigma,N \rangle$ \textbf{from} $\pks$ 
     
   $\wedge$
$(\tid.\pkb \in \tid.F.\owners)$ $\wedge$
$ (\tid.\pks = \pks)$

 $\wedge$  
     $\tid.F \not\in \mathit{settle}$ $\wedge \,$ 
$(\tid.F \not\in \mathit{validated\_fund})$ 

$\wedge$      
     \Call{VerifySig}{$\tid.\pkb,\tid||h_s||\textproc{H}(v||N)),\sigma$} }
     
     \State $\mathit{validated\_fund} = \mathit{validated\_fund} \cup F$
     \State $\mathit{validated\_txs} = \mathit{validated\_txs} \cup \langle \tid, h_s,\sigma, N \rangle$
     \State \textbf{send} $\langle \valid, \tid ,h_s, \sign_\skv(\tid||h_s) \rangle$ \textbf{to} $\pks$  
     \Else
     \State \textbf{send} $\langle \invalid, \tid \rangle$
        \textbf{to} $\pks$      
    \EndIf
    \end{algorithmic}
\end{algorithm}

\subsubsection{Seller's Code}
The seller's \textproc{SellerFractionalSpending()} function is executed in response to a PAY message from the buyer and 
takes the $\tid$ of the PAY message as argument. The seller starts by calling \textproc{SelectQuorum()} (see Section~\ref{sec:quorum-select-proofs}) using the buyer's transaction identifier and a randomly generated nonce $N_s$. The call  returns a quorum of validators $Q_\tid$ represented as a vector  $[v_i]$.
The seller generates $m$ $r$-bit nonces to generate a vector of commitments  
$[c_i] = [H(v_i||N_i)]$ to validators identities, which it sends to the buyer, and waits to receive from the buyer the signed commitments. Then the seller sends to each validator a validation request that includes the transaction identifier $\tid$, the signature $\sigma_i$ (which should be equal to 
$\sign_\skb(\tid||c_i)$), and the nonce $N_i$ used in generating the commitment. It then waits to receive replies from  $m-(1+\mu)p_fm$ validators and checks if $(1-\alpha) m$ validators validated the transaction. If so, the transaction is validated and the replies from validators that validated the transaction constitute a {\em certificate} of validation for the transaction. In the code, we do not explicitly represent the resulting fund, but we note  that the transaction, and therefore the fund that it creates, is specified by
the transaction identifier $\tid = \langle F, \pkb, \pks \rangle$ and $hs$, the seller's commitment to the nonce $N_s$.  

\subsubsection{Validator's Code}
The validator checks  the validity of the request and sends the result to the seller. A request is valid if: the fund of the payment is not settled; the seller of the transaction is the one making the request; the buyer of in the transaction is an owner of the fund of the transaction; the fund specified by the transaction has not been previously validated by the validator; the signature 
$\sigma$ in the request is a valid signature for $\tid||h_s||H(v||N)$ by the buyer $\tid.\pkb$ specified in the request. The validation of the signature ensures that the seller committed to the validator's identity for the specific transaction.
If the request is valid, the validator adds the fund to the set of validated funds and the transaction to the list of validated transactions. The list of validated funds is needed to ensure that a validator validates only one fractional payment from a given fund. It is used during settlement of buyer's fund $F$ to ensure that all fractional spending from $F$ will be accounted for. That is why $\sigma$, $c$ and $N$, which are needed to validate a transaction, are stored together with $\tid$. It is important to note  that a validator does not learn the identities of other validators validating a payment request. 

\subsection{Settlement Protocols}\label{sec:settle}
As we explained in the overview, during settlement we need to prevent the adversary from learning the identities of the validators prematurely and corrupting them. 
The settlement algorithms use a \textproc{Propagate()} protocol that allows an honest party to propagate information without the adversary having the ability to suppress the information being propagated. We outline the \textproc{Propagate()} protocol first (the details and code can be found in the Appendix), then we present the settlement protocols.  

\subsubsection{Propagating Information}
The identities of individual validators  involved 
in validating a particular transaction are kept secret by the validators themselves (if honest) or by the seller until settlement time. At settlement time, the identities of the validators are revealed, so we need to ensure that when the adversary learns the identities of the validators, it will not be able to erase the information about the transaction. The \textproc{Propagate()} protocol allows a  party (seller or validator) to send a message to all validators so that the adversary would either
have to corrupt the sending party to learn the message or, if the adversary learns the message without corrupting the party, then all but $2f$ honest validators are also guaranteed to learn the message. The \textproc{Propagate()} protocol is implemented using secret sharing and reconstruction.  
To distinguish between messages propagated by different calls to \textproc{Propagate()}, a unique  nonce $N_{prop}$ is generated for the call and provided as a second argument to \textproc{Propagate()}. When a message from client $c$ is finally propagated to a server, the server will have $\mathit{message}[c,N_{prop}]$ equal to the propagated message.

\subsubsection{Seller's Settlement and Corresponding Validation} 
The protocols for the seller's settlement and corresponding validation are shown in Algorithm \ref{seller-settlement}  and Algorithm \ref{validator-seller-settlement}. The seller  propagates a settlement request to validators. The fund is fully specified by $\tid = \langle F, \pks,\pkb \rangle$ and the nonce $N_s$ of the fractional spending transactions that created the fund. 
The certificate of the fund consists of the set of $\mathit{payment\_witnesses}$ that validated the fractional spending transaction. The information $\tid, N_s, \mathit{payment\_witnesses}$ enables a validator of the settlement transaction to recalculate the quorum used in validating the transaction and to verify that the witnesses
are members of that quorum. The validator checks that the quorum and validations provided by the seller are correct and that the validation for the transaction are received from a large enough subset of validators. In addition, to avoid double spending, the validator checks that either it has no record of a settlement transaction for $F$ ($F \not\in \mathit{settle}$) or that the payment transaction being settled, represented as $(\tid, h_s = \Call{H}{N_s})$ was added to $\mathit{transactions}[F]$ which is the set of payment transactions from $F$ (calculated when the buyer settles $F$). 
As shown in the proofs in the Appendix, if the buyer settles $F$, every payment to an honest seller will be added to  
$\mathit{transactions}[F]$ which ensures that this condition will be satisfied for honest sellers.
If all the information checks out, the validator sends a validation for the settlement. The seller waits until it receives $n-f$ validations. 
These $n-f$ validations guarantee that only one fund can result from settling a partially validated fund. 

\begin{algorithm}[t]
    \caption{Seller Fund Settlement}
    \label{seller-settlement}
    \begin{algorithmic}[1] 
        \Procedure{SellerSettle}{ $\tid = \langle  F, \pkb,\pks \rangle$, $N_s$,
        
        \hspace*{44mm}$\mathit{payment\_witnesses}$ }
       \State $N_{settle} \leftarrow \{0,1\}^r$
     \State \Call{Propagate}{$\langle \tid, N_s, \mathit{payment\_witnesses},\text{SETTLE}\rangle$, $N_{settle}$}  
        \Repeat 
            \If{$\rcvd (\textbf{valid},\sigma, N_{settle} )$ \textbf{ from } $v$}
            \State $\mathit{Id}_1 = \Call{H}{\tx||N_s||\mbox{PAY}}$;
            \hspace{3ex} $\mathit{Id}_2 = \Call{H}{\mathit{Id}_1||\mbox{SETTLE}}$
            \State $F'.\fid = Id_2 \,;\,  F'.\fbl =  F.\fbl/k'_2 \,;\, F'.\owners = \{\pks\}$
            \If{$\sigma=\sign_v(\langle F'.\fid, F'.\fbl, F'.\owners\rangle$}
                \State $\mathit{settle\_witnesses} = \mathit{settle\_witnesses} \cup  \{(v,\sigma)\}$
            \EndIf
            \EndIf
        \Until{$|\mathit{settle\_witnesses}| \geq n-f$}
        \State $F'.\fcert = \mathit{settle\_witnesses}$
        \State $\textbf{return } (\langle F'\rangle)$
       \EndProcedure
    \end{algorithmic}
\end{algorithm}

\begin{algorithm}[t]
    \caption{Code of Validator $v$ for Seller Fund Settlement}
    \label{validator-seller-settlement}
    \begin{algorithmic}[1] 
           \If{
           $\mathit{message}[\pks,N_{settle}] = \langle \tid, N_s, \mathit{payment\_witnesses},$
           
           \hspace*{58mm}$\text{SETTLE}\rangle$  
           
        $\,\,\,\,\,\wedge$ 
        $Q = \Call{SelectQuorum}{\tid, N_s}$ 
       
        $\,\,\,\,\,\wedge$  
		$((\tid, \Call{H}{N_s}) \in \mathit{transactions}[F] \vee F \not\in \mathit{settle})$ 
  
       $\,\,\,\,\,\wedge$ 
       $\Call{ValidatedTransaction}{\tid,\mathit{payment\_witnesses}}$ 
       
       $\,\,\,\,\,\wedge$ 
       $Q' = \{ q\,:\, (\mathit{resp},q) \in \mathit{payment\_witnesses}\} \subseteq Q$ 

       $\,\,\,\,\,\wedge$ 
        $|Q'| \geq m-m\mu p_f$ }
            \State $\mathit{transactions}[F] = \mathit{transactions}[F] \cup \{((\tid, \Call{H}{N_s})\}$
       		\State $\mathit{Id}_1 = \Call{H}{\tx||N_s||\mbox{PAY}}$
            \State $\mathit{Id}_2 = \Call{H}{\mathit{Id}_1||\mbox{SETTLE}}$
            \State $F'.\fid = Id_2 \,;\,  F'.\fbl =  F.\fbl/k'_2 \,;\, F'.\owners = \{\pks\}$
            \State \send $(\textbf{valid}, \sign_v(\langle F'.\fid, F'.\fbl, F'.\owners\rangle, N_{settle})$ to $\pks$
            \EndIf
           
    \end{algorithmic}
\end{algorithm}
\subsubsection{Buyer's Settlement and Corresponding Validation} To settle a fund $F$, we need to make sure that all fractional spending from $F$ are deducted from the settled balance. We assume that there is only one settlement transaction that is executed for the fund and that the buyer is engaged in its execution. The buyer starts by sending a settlement request to all validators. Every validator that receives the request {\em propagates} information about any payments from $F$ that it is aware of (there can only be at most one such payment per honest validator). A validator will provide proof that it witnessed a payment or states that it did not witness any payment. The proof of a witnessed payment is $(\tid, \sigma, N_h)$ where $\sigma = \sign_{\pkb}(H(\tid||N_h))$ is the blind signature and $N_h$ is the blinding factor. Validators wait until $n-f$ different messages about payments from $F$ are propagated. Every validator then counts the number of different payment transactions from $F$ that it heard about, calculates the resulting balance after subtracting the amounts for those payments and send the buyer a signed balance for the settlement fund. The buyer waits until it receives $n-2f$ signatures for the same balance which will form the set of witnesses for the settlement fund resulting from $F$. The code can be found in the Appendix.

\subsection{Performance}
A simple inspection of the code shows that validating fractional payment transactions requires a linear number of messages because the only communication is between the seller and the validators, which is optimal. Also, the latency is one roundtrip message delay between the seller and the validators.
Settlement requires a quadratic number of messages. The settlement latency is higher requiring three rounds of message exchanges between validators. 

\section{Random Quorum Selection}
\label{sec:quorum-select-proofs}
The $(k_1,k_2)$-quorum systems definition assumes that a quorum can be chosen randomly, so we need to specify how quorums can be chosen randomly by the seller and how to prevent corrupt sellers from fixing the membership of the chosen quorum.

\begin{algorithm}[!t]
    \caption{Random selection of a quorum of size $m$: Seller's's Code ($\mathit{pks}$)}
    \label{random-quorum-selection}
    \begin{algorithmic}[1] 
        \Function{SelectQuorum}{$\langle F, \pkb , \pks \rangle , N_s$}

    \State $T_s = \langle F, \pkb , \pks \rangle ||N_s$ \color{blue}\Comment{seller's transaction identifier}\color{black}\label{Tunique}
       \State $h = H(T_s)$ 
     \State  $Q = \{ \}$ ; $j = 1$ ; 
     \While{$|Q|< m$}
           \If{$\mathit{Server}(H(h||j)) \not\in Q$} \color{blue}\Comment{if not already selected}\color{black}\label{line:star}
        \State $Q = Q \cup \{\mathit{Server}(H(h||j))\}$ \color{blue}\Comment{add server to quorum}\color{black}
              \EndIf
        \State $j = j+1$
     \EndWhile
     \State \Return $Q$ 
	           \EndFunction
    \end{algorithmic}\label{alg:quorum-select}
\end{algorithm}

The protocol (Algorithm~\ref{alg:quorum-select})  allows the seller to choose a quorum of $m$ validators and ensures that if the seller is honest, the quorum of validators is chosen uniformly at random and is not known to the adversary or to the buyer. 
The arguments to the algorithm are a buyer-provided transaction identifier $\tid = \langle F, \pkb, \pks \rangle$ that ties the chosen quorum to the two parties and the specified fund, and a seller-generated random $r$-bit string $N_s$,  where $r$ is a security parameter. The transaction's identifier is concatenated  (denoted with $||$) with $N_s$   
to obtain a seller transaction identifier $T_s$, which is 
used as a seed for quorum selection. This seed is 
guaranteed to be unique \whp if the seller is not corrupt. 
The goal is to select $m$ different servers. The seller uses the seed $h$ concatenated with an index $j$ to select the validator $\mathit{Server}(H(h||j))$, where $\mathit{Server}$ is a function that maps the output of $H$ to server identifiers. Since the number of possible validators is $n$, it is possible that the same validator is chosen twice for different values of $j$, so the seller tries successive values of $j$ until $m$ distinct validators are chosen.   
The algorithm returns the identities of the selected validators.  
This information is all that is needed to verify later that a particular quorum of validators was properly chosen according to the protocol. In fact, the set of chosen validators is a deterministic function of the protocol arguments, 
but this requires knowledge of $N_s$ without which an adversary cannot guess the identities of the validators. 

 It is important to point out that the randomness of the quorum is ensured by using $\langle F, \pkb , \pks \rangle ||N_s$ as a seed for quorum selection and that a corrupt seller cannot reuse an old seed to double spend because this would result in the same seller transaction identifier $T_s$.
One final consideration is ensuring random selection for a  corrupt seller that attempts to run the algorithm multiple times in the hope of maximizing the number of previously corrupted validators in the chosen quorum or to present a different $N_s$ at validation time to ensure that the quorum contains a large number of corrupt validators. Since the quorum system is asynchronous, \whp there are no more than $(1+\mu)p_fm$ previously corrupted validators in a randomly chosen quorum. If the seller wants to increase the number to $(1+2\mu)p_fm$, for example, the seller should make an exponential number of attempts (exponential in $\mu p_fm$) to choose a quorum. So, we assume that the quorum system is such that \whp a randomly chosen quorum does not contain more than $(1+\mu/2)p_fm$ previously corrupted validators, and that $(\mu/2)p_fm/2$ is large enough so that \whp the computationally bounded corrupt seller cannot chose a quorum with more than $(1+\mu)p_fm$ previously corrupted validators.

The quorum selection protocol satisfies the following properties:

\newtheorem*{L3}{Lemma~\ref{lem:quorum-random}}
\begin{lemma}\label{lem:quorum-random}
If the seller is correct, the quorum of validators is chosen uniformly at random.
\end{lemma}

\newtheorem*{L4}{Lemma~\ref{lem:unknown-validators}}
\begin{lemma}\label{lem:unknown-validators}
If the seller is correct, the identities of the correct validators in the chosen quorum are not known to the adversary or to the buyer.
\end{lemma}

\begin{lemma}\label{lem:corrupt-validator-ratio}
The probability that a randomly selected quorum of size $m$ contains $(1 + \mu)p_fm$ previously corrupt validators, $0 < \mu < 1$ is at most $e^{-\mu^2 \times p_f m/(2+\mu)}$.
\end{lemma}

The following Lemma follows directly from Lemma~\ref{lem:corrupt-validator-ratio}. 
\newtheorem*{L5}{Lemma~\ref{lem:quorum-corrupt-seller}}
\begin{lemma}\label{lem:quorum-corrupt-seller}
For $0 < \mu < 1$, if a seller chooses $K$ quorums, of size $m$ each, at random, then with probability at most $Ke^{-\mu^2 \times p_f m/(2+\mu)}$, every chosen quorums has no more than $(1+\mu)p_fm$ previously corrupt validators.
\end{lemma}

Proofs can be found in the Appendix.

\section{Impossibility Proofs}\label{sec:impossible}
We show that a deterministic solution cannot satisfy all the problem's requirements and that no solution that uses less than full quorums can spend the whole amount and avoid double spending.

\begin{lemma}\label{lemma:nodeterministic}
If all payments are partially certified (certified by quorums of size less than $f+1$), then a deterministic solution to the Fractional Spending problem cannot guarantee progress \whp for correct buyers and sellers.
\end{lemma}
\begin{proof}
In a deterministic solution, the identities of the validators for a given transaction is solely determined by the buyer and seller identifiers $\pkb$ and $\pks$, the history of transactions at the buyer and seller and the initial states of the buyers and sellers. Consider an execution in which all payments are made by one buyer and assume that the buyer pays the seller with a fractional spending transaction. The identity of the validators that will be used by the seller can be calculated by the buyer because it has all the information needed to do that. If the buyer is corrupt, the adversary can learn the identities of the validators for the transaction and corrupts all of them because they number less than $f+1$. The seller can then be prevented from successfully settling the fund resulting from the transaction.
\end{proof}

\begin{lemma}\label{lemma:notfullbalance}
If all quorums used in validation have size less than $f+1$, then it is not possible to prevent double spending and allow the spending of the whole balance of a fund.
\end{lemma}
\begin{proof}
We consider an execution in which a buyer spends the whole balance of a fund $F$ by executing $s$ fractional spending transactions, $T_1,T_2,\ldots,T_s$, for some $s$, to honest sellers who validate and accept the payments. No validators are corrupted during the executions of these $s$ fractional spending transactions. After these fractional spending transactions are validated, the adversary corrupts the buyer who issues a fractional spending transaction $T$ to a corrupt seller. To validate the transaction $T$, less than a total of $f+1$ validators are contacted. In a general solution, some of the validators for a transaction are contacted directly by the seller and buyer and some of them could be contacted indirectly by other validators. Since the adversary controls the buyer and seller, it knows the identities of the validators that are directly contacted and can corrupt them because there are at most $f$ validators involved in validating the transaction. In turn, and for the same reason, this allows the adversary to learn the identities and corrupt any validators that are used to validate the transaction whether they are directly or indirectly contacted by the seller and buyer. Since all the validators for transaction $T$ can be corrupted by the adversary, the transaction can be successfully validated and successfully settled.   
Since the transactions $T_1,T_2,\ldots,T_s$ are to honest sellers, according to the problem requirements, the sellers for these transactions should be able to successfully settle the funds resulting from them, totalling the whole initial balance of $F$. So, the total of all transactions that are settled is equal to initial balance of $F$ plus the amount for the transaction $T$, a total spending that exceeds the original balance $F$, a double spending.
\end{proof}

\section{Conclusion}\label{sec:conclude}
We introduced a novel problem called the fractional spending problem and a novel quorum system that enables our solution.  We showed how to use the quorum system to execute payment transactions with less than $f$ validations per transaction. By carefully considering lower bound and upper bounds in the intersection properties of the quorum system, we are able to allow up to $s_1=k_1$ non-interfering payments in parallel and prevent double spending by limiting the number of such concurrent transactions. We solved the problem in an asynchronous system under an adaptive adversary. 
The work leaves a number of important questions unanswered. The construction we proposed for the quorum works for large values of $n$. It is not clear how to modify the construction to work for smaller values of $n$ and larger values of $f$ relative to $n$, or how to reduce the difference between $k_1$ and $k_2$. This is a subject for future work. Also, the current solution requires validators to maintain a history of past transactions which results in space complexity proportional to the number of transactions. Finally, the solution requires the seller to interact with the buyer to blindly validate the list of validators. Eliminating this interaction is desirable to improve the system efficiency. This is subject for future work.
\bibliography{biblio}
\newpage
\renewcommand{\thesection}{\Roman{section}} 
\renewcommand{\thesubsection}{\thesection.\Roman{subsection}}
\setcounter{section}{0}
\section*{Appendix}
This appendix contains the following:
\begin{enumerate}




\item  {\bf Propagating Information} (Section~\ref{sec:propagate}): This section includes the protocols for propagating information with their correctness proof. 

\item  {\bf Fractional Spending and Settlement} (Section~\ref{ap:settle}): This section provides code for the buyer settlement protocol and the proofs that our solution for the payment and settlement satisfy the problem requirements.

\item  {\bf $(k_1-k_2)$-Quorum Systems Properties} (Section~\ref{sec:k1k2-proofs}): This section provides the proofs for the lemmas relating to $(k_1-k_2)$-quorum systems that are listed in the main text.

\end{enumerate}

\section{Propagating Information: Protocols and Properties}\label{sec:propagate}

As we discussed in the protocol overview, the adversary can 
corrupt a validator if it learns that the validator validated a 
particular transaction. The identities of validators that are involved 
in validation are kept secret by the validators themselves (if honest) or the seller until settlement time, at which time they need to be divulged.  The \textproc{Propagate()} protocol allows a  party (seller or validator) to send secret information to all validators so that the adversary would either
have to corrupt the party to learn the information or, if the adversary learns the information without corrupting the party, then all but $f$ honest validators ($n-2f$ validators) are also guaranteed to learn the information. The \textproc{Propagate()} protocol is implemented using secret sharing and reconstruction. The code is given in Algorithms~\ref{propagate-c} and~\ref{propagate-v}.

\begin{algorithm}[t]
    \caption{Propagating Information: Client $c$'s code}
    \label{propagate-c}
    \begin{algorithmic}[1] 
   \Procedure{PropagateClient}{$\mathit{message},N_{prop}$}
     \State $[s_i] = $ \Call{SecretShare}{$\mathit{message}$, $n$, $f$} 
     \For{$i:= 1 \to n$}             
      \State \send $(\mbox{SHARE}, N_{prop},  s_i,\sign_c(s_i||v_i||N_{prop})) \textbf{ to }  v_i$
        \EndFor

\vspace{1ex}
		\State $\mathit{acks}[N_{prop}] = \emptyset$
        \Repeat 
        \If{$\rcvd (\mbox{SHARE\_ACK}, N_{prop}) \textbf{ from } v$}
                \State $\mathit{acks}[N_{prop}] = \mathit{acks}[N_{prop}] \cup  \{v\}$
            \EndIf
        \Until{$|\mathit{acks}[N_{prop}]| \geq n-f$}
        \State \send $(\mbox{RECONSTRUCT}, N_{prop}) \textbf{ to } \validators$ 
        
        \vspace{1ex}
        \State $\mathit{reconstructed}[N_{prop}] = \emptyset$
        \Repeat 
        \If{$\rcvd (\mbox{RECONSTRUCTED}, \mathit{message}, N_{prop}) \textbf{ from } v$}
                \State $\mathit{reconstructed}[N_{prop}] = \mathit{reconstructed}[N_{prop}] \cup  \{v\}$
            \EndIf
        \Until{$|\mathit{reconstructed}[N_{prop}]| \geq n-f$}
        
       \EndProcedure
    \end{algorithmic}
\end{algorithm}

\begin{algorithm}[t]
    \caption{Propagating Information: validator $v$'s code. The client $c$ in the code can be either a seller invoking the client side of information propagation to settle a fund resulting from fractional payment transaction or a validator invoking propagation information to handle a buyer's settlement transaction}\label{propagate-v}
    \begin{algorithmic}[1] 
        \Procedure{PropagateServer}{$N_{prop}$}
        \State {\bf repeat}
        \State\hspace{3ex} \upon \rcpt $(\mbox{SHARE},N_{prop}, s,\sigma = \sign_\pks(s||v||N_{prop})) \textbf{ from } c$: 
			\State\hspace{6ex} \send $(\mbox{SHARE\_ACK}, N_{prop}) \textbf{ to } c$
			\State\hspace{6ex} $\mathit{share}[c,N_{prop}] = s,\sigma$
     
        \vspace{1ex}
        \State\hspace{3ex} \upon \rcpt $(\mbox{RECONSTRUCT}, N_{prop}) \textbf{ from } c$:
     		 \State\hspace{6ex} \send $(\mbox{FORWARD}, c, \mathit{share}[c,N_{prop}], N_{prop}) \textbf{ to } \validators$
		     \State\hspace{6ex} $\mathit{shares}[c,N_{prop}] = \emptyset$, $\mathit{Forwarded}[c,N_{prop}] = \emptyset$

	             \vspace{1ex}
	   			
			\vspace{1ex}
        \State\hspace{3ex} \upon \rcpt $(\mbox{FORWARD},c,s, \sigma,  N_{prop})$ \textbf{ from } $v$:
       
        	\State\hspace{6ex} \bif $\sigma = \sign_c(s||v||N_{prop}) \wedge v \notin \mathit{Forwarded}[c,N_{prop}]$ {\bf then}

                \State\hspace{9ex} $\mathit{shares}[c,N_{prop}] = \mathit{shares}[c,N_{prop}] \cup  \{ (s,\sigma)\}$                				
                \State\hspace{9ex} $\mathit{Forwarded}[c,N_{prop}] = \mathit{Forwarded}[c,N_{prop}] \cup  \{ v \}$
        \State \hspace{9ex} \bif {$|\mathit{shares}[c,N_{prop}]| \geq f+1$} {\bf then}
			 \State\hspace{12ex} $\mathit{message}[c,N_{prop}] =$ 
			 				\Call{ReconstructSecret}{$\mathit{shares}[c,N_{prop}]$}
         	\State\hspace{12ex} $\send (\mbox{RECONSTRUCTED}, \mathit{message}[c,N_{prop}], 
								N_{prop}) \textbf{ to }  c \mbox{ and } \validators$ 

       \State {\bf until} $(n-f)$ RECONSTRUCT messages received
       \EndProcedure
    \end{algorithmic}
\end{algorithm}

Propagating information is relatively simple in our setting and is easier than the asynchronous verifiable secret sharing problem~\cite{cachin2002asynchronous} (a solution to which can also be used, but would be an overkill). The client protocol takes as input two arguments: the {\em message} being propagated and a unique nonce $N_{prop}$ to distinguish between different calls to the propagate protocol. The client $c$ who propagates the information creates shares of the information
that needs to be propagated so that any $f+1$ out of $n$ shares can be used to reconstruct the {\em message}. Depending on the settlement protocol, $c$ is either a seller or a validator.  Each share $s_i$ is addressed to a particular validator $v_i$. In addition to the share $s_i$, $v_i$ receives from the client a signature $\sigma_i = \sign_c(s_i||v_i||N_{prop})$ that $v_i$ can later use to prove that $s_i$ is a properly received share. After the shares are distributed and $n-f > 2f+1$ validators acknowledge receiving the shares,
the client sends a RECONSTRUCT message so that the validators can reconstruct the {\em message} from the shares. Validators broadcast their shares and collect shares that they authenticate until $f+1$ authenticated shares are collected. At that point, the {\em message} is reconstructed using the collected shares. It is important to note here that all the shares used in the reconstruction are validated  to be form the client (line 9 of validator code), so if the client is honest, the $f+1$ shares will all be valid even if they are received from corrupt validators. If the client is not honest, we don't care about the value that is reconstructed.
A validator stops participating in the propagation protocol for a particular $N_{prop}$ when $n-f$ different validators announce that they have reconstructed the messages. Of those validators, $n-2f$ must be honest.

The algorithms guarantees that if the client calling \textproc{Propagate($\mathit{message}$)} is honest, $n-2f$ honest validators will receive $\mathit{message}$. If the client propagating the shares is not honest, there are no guarantees as to what is reconstructed and the reconstruction might even fail, but that does not affect the algorithms that use the \textproc{Propagate()} function.

\begin{lemma}
If an honest client calls \textproc{Propagate($m$)}, $n-2f$ honest validators will be guaranteed to receive $m$. 
\end{lemma}
\begin{proof}(sketch)
If the client is honest, when a validator receives a RECONSTRUCT message, $n-f$ validators must have already received the SHARE message. Of these validators $n-2f \geq 2f+1$ validators are honest and will each send FORWARD message containing its share to every other honest validator who will be able to reconstruct the message.
Validators stop their execution when they receive $n-f$ RECONSTRUCT messages from other validators. Of these $n-2f$ are from honest validators.
\end{proof}

\begin{lemma}
Let $N$ be a random value that is generated by a client $c$ such that no value that is directly or indirectly dependent on $N$ is sent to any validator with the exception of  $\textproc{H}(N)$. If the adversary learns $N$ \whp, then \whp the adversary must have corrupted the client $c$.
\end{lemma}
\begin{proof}(sketch)
The proof is by induction on the number of messages sent by $c$. If $c$ sends no messages, it should be clear that the adversary cannot learn $N$ without corrupting $c$. Assume that the lemma holds for the first $i$ messages sent by $c$ and consider the $i+1$'th message sent by $c$. The $i+1$'th message either contains no value dependent on $N$ or contains $\textproc{H}(N)$ and other values that are not dependent on $N$. In the first case, the adversary clearly doesn't learn anything about $N$. In the second case, given the adversary's bounded computational power, it cannot learn the value of $N$ except with negligible probability. It follows that the chain must be a 0-length chain and the adversary learns $N$ by corrupting $c$.
\end{proof}

\begin{lemma}
Let $N$ be a random value that is generated by a client $c$ such that no value that is directly or indirectly dependent on $N$ is sent to any validator with the exception of  $\textproc{H}(N)$ or by executing $\textproc{Propagate}(m)$ for a message $m$ that depends on $N$. If the adversary learns $N$ \whp, then \whp the adversary must have corrupted the client $c$ or $n-2f$ honest validators are guaranteed to learn $m$.
\end{lemma}
\begin{proof}(sketch)
Since each share generated by the secret sharing of \textproc{Propagate}($m$) is independent of $N$, by an argument similar to that in the previous lemma, the adversary cannot learn $N$ \whp if it does not learn $f+1$ shares. If the adversary learns $N$, then it must have corrupted a validator  that received $f+1$ valid shares. One of these shares must be from an honest validator that received a RECONSTRUCT message. It follows that $n-f$ validators will receive RECONSTRUCT messages because the client stays active until reconstruction is confirmed by $n-f$ validators.  Honest validators that receive RECONSTRUCT messages send FORWARD messages to all other validators who will be able to reconstruct $m$. Validators participate in the protocol until $n-f$ validators announce that they reconstructed the message at which point $n-2f$ honest validators must have reconstructed the message.
\end{proof}

It might seem a little strange that even though every honest validator sends messages to every other honest validator, we only guarantee that $n-2f$ honest validators will receive the message if the adversary doesn't corrupt the client. The reason is that even though honest validators might be guaranteed to eventually receive messages sent by other honest validators, due to system asynchrony, we need to make a statement about what holds when the client finishes executing \textproc{Propagate}($m$). 

\section{Fractional Spending and Settlement Proofs}\label{ap:settle}

The protocol for buyer fund settlement is shown in Algorithms~\ref{buyer-settlement}
and~\ref{v-buyer-settlement}. Remember that we assume that only one settlement transaction will be invoked for a given fund. In order for the execution to complete, at least one of the owner should be engaged in its execution. The rest of this section presents the correctness proof of the solution for fractional spending and settlement protocols.
.

\begin{algorithm}[t]
    \caption{Buyer Fund Settlement}
    \label{buyer-settlement}
  \begin{algorithmic}[1] 
    \Procedure{BuyerSettle}{$F$}
        \State \send $(\mbox{SETTLE},F) \textbf{ to }  \validators$
         \Repeat 
            \If{$\rcvd \mathit{resp} \textbf{ from } q_i$} \color{blue}\Comment{Wait for replies from }\color{black}
                \State $\mathit{replies} = \mathit{replies} \cup \mathit{resp} $ \color{blue}\Comment{validators until enough }\color{black}
            \EndIf
        \Until{$\exists F'\,:\,|\{\mathit{resp}\,:\, \mathit{resp}\in \mathit{replies} \wedge \mathit{resp} = F'\}| = n-2f$} \color{blue}\Comment{identical replies are received}\color{black}

        \State \textbf{return } $F'$
    \EndProcedure
  \end{algorithmic}
\end{algorithm}

\begin{algorithm}[!ht]
    \caption{Code of Validator $v$ for Buyer Fund Settlement}
    \label{v-buyer-settlement}
  \begin{algorithmic}[1] 
    \Procedure{ValidatorBuyerSettle}{$F$, $\pkb$}
    \If{$\rcvd (\mbox{SETTLE},F) \textbf{ from } 
    		\pkb \wedge \pkb \in F.\owners$} 
    	\State $N_{prop} \leftarrow \{0,1\}^n$	
				\color{blue}
				\Comment{different validators choose different $N_{prop}$ values} 
				\color{black}
        \State $value[F,N_{prop}] = \bot$ 
        		\color{blue}
				\Comment{$F$ will have multiple entries, one for each such value} 
				\color{black}

        \vspace{1ex}
        \color{blue}\Comment{If $v$ is witness to payment from $F$, propagate payment info} \color{black}
        \If{$\exists \,(\tid, h_s,\sigma, N) \in \mathit{validated\_transactions}$ $\wedge$
        
        \hspace{6ex}	$\tid.F = F$ $\wedge$ $pkb \in \tid.\owners$} 
          \State  \Call{Propagate}{$\langle \langle F,\mbox{SETTLE}\rangle \rangle$, $v$ , 
          			$(\tid,h_s,\sigma, N), N_{prop}$}

        \Else \color{blue}\Comment{otherwise propagate that $v$ is not witness to any payment from $F$} \color{black}
          \State \Call{Propagate}{$\langle \langle F,\mbox{SETTLE}\rangle \rangle$, $v$ , 
          	$\sign_v(F||none),N_{prop}$} 
	        \EndIf
    \EndIf

    \vspace{1ex}
    	\color{blue}\Comment{Update } known $\mathit{transactions}[F]$ based on propagated information\color{black}
    \If{$\mathit{message}[F,N_{prop}] = \langle \langle F,\mbox{SETTLE}\rangle \rangle$, $v$ , 
          			$(\tid,h_s,\sigma, N), N_{prop}\rangle$}
\State $\mathit{transactions}[F] = \mathit{transactions}[F] \cup \{(\tid,h_s)\}$
     \EndIf
     
	\vspace{1ex}
	\color{blue}\Comment{Forward any witness or non-witness information received} \color{black}
    \If{$\mathit{message}[F,N_{prop}] = (\langle F,\mbox{SETTLE}\rangle, v', *)$}
       \State $\send (F, \mbox{SETTLE}, value[F,N_{prop}] )$ to all validators
     \EndIf

	\vspace{1ex}
    \If{$\rcvd (\langle F,\mbox{SETTLE}\rangle, v', *)$ 	
		$\wedge$ $v' \not\in \mathit{settle}\_\mathit{validators}[F]$}
         \State $\mathit{settle}\_\mathit{validators}[F] = 
         		\mathit{settle}\_\mathit{validators}[F] \cup v'$
    \EndIf
     
    \vspace{1ex}   
    \If{$\rcvd (F, \mbox{SETTLE}, v' , (\tid,\sigma,N))$ 
    		$\wedge$ $v' \not\in \mathit{settle}\_\mathit{validators}[F]$ $\wedge$
		
		\hspace{1ex} 
			$\Call{ValidPayment}{F, (\tid,h_s,\sigma,N})$}                   
			\State $\mathit{payments}[F] = \mathit{payments}[F] \cup 
											(\tid, h_s,\sigma, N)$
    \EndIf

       	\vspace{1ex}
    \If{$|\mathit{settle}\_\mathit{validators}[F]| = n-f$}
                \ForAll{$(\tid, h_s,\sigma, N) \in \mathit{payments}[F]\,:\, 	
                		\Call{ValidPayment}{\tid,h_s,\sigma,N}$}
                \State $\mathit{transactions}[F] = \mathit{transactions}[F] \cup \{(\tid,h_s)\}$
               
                \State $\mathit{settle} = \mathit{settle} \cup \{F\}$ 
                \If{$|\mathit{transactions}[F]| > s_1$}        		\color{blue}
				\Comment{if many transactions were issued from $F$} 
				\color{black}
                	\State {\bf abort} \color{blue}
				\Comment{abort} 
				\color{black}
				\EndIf
        		\color{blue}
		
		\vspace{1ex}
				\Comment{Add $F$ to funds with SETTLE transactions} 
				\color{black}
				 \EndFor
          		\State $\mathit{new}\_\mathit{balance} = F.\fbl - |\mathit{transactions}[F]|\times 1/s_2$
         		\State $F'.\tid = H(F.\tid)\,;\, F'.\fbl = \mathit{new}\_\mathit{balance}\,;\, F'.\owners = F.\owners$
        \State $\send (F',\sign_v(F')) \textbf{ to } F'.\owners$
     \EndIf
    \EndProcedure
  \end{algorithmic}
\end{algorithm}

\begin{lemma}[Transactions to honest sellers are accounted for]
Let $T_{s} = (\tx, N_s)$ be a fractional spending transactions for which the seller is honest. If the buyer executes $(\mbox{SETTLE},F)$, then $n-2f$ honest validators will have $(\tx, h_s = \textproc{H}(N_s)) \in \mathit{transactions}[F]$. 
\end{lemma}
\begin{proof}(sketch)
If the seller already executed a settlement transaction for the fund resulting from $T_s$, then it should have received validations from $n-f$ validators each of which adds $(\tx, h_s = \textproc{H}(N_s))$ to $\mathit{transactions}[F]$. So, we assume that the seller did not execute a settlement transaction for the fund resulting from $T_s$ and did not divulge $N_s$. It follows that the identities of the validators of $T_s$ are not known to the adversary who will not be able to selectively delay messages from all the validators of $T_s$ (see discussion in Section~\ref{sec:model} regarding adversary's control of asynchronous communication), which means that some of the honest validators of $T_s$ will succeed in propagating their information to $n-2f$ honest  validators that add 
$(\tx, h_s)$ to $\mathit{transactions}[F]$. 
\end{proof}

\begin{lemma} All partially certified funds with honest owners can be settled successfully. 
\end{lemma}
\begin{proof}(sketch)
Consider a partially certified fund $F_s$ resulting from a fractional spending transaction $\tid$ from a fund $F$. If the owner (seller) of $F_s$ is honest, then when the payment was made, the seller selected a quorum according to the quorum selection protocol and then got the payment validated by a sufficient number of honest validators from the selected quorum. The identities of these validators are not known to the adversary before the owner settles the fund. Also, each of these honest validator received the validation request from the seller before receiving a settlement request for $F$ because one of the conditions for validating a payment request is for the validator not to have $F$ is the set of funds for which there is a SETTLE transaction ($F.\tid \not\in \mathit{settle}$). Every honest validator that receives the settlement request for $F_s$ will either have $F \not\in \mathit{settle}$ or  $F \in \mathit{settle}$. If $F \in \mathit{settle}$, by the previous lemma, the validator must have added  $F_s$ to $\mathit{transactions}[F]$.
The validation of the settlement of $F_s$ will go through in both cases because the other conditions for validating a seller's settlement transaction will hold because the honest seller provides a correctly selected quorum and a large enough set of witnesses that validates the fractional spending transactions that created $F_s$.
\end{proof}

\begin{lemma}
Settlement for a partially certified funds equals payment amount: If $F''$ is a fully certified fund resulting from executing $(F',\mbox{SETTLE})$ for partially validated fund $F'$, then $F''.\fbl = F'.\fbl$.  
\end{lemma}
\begin{proof}(sketch)
This follows immediately from the code. When a fund $F'$ of a fractional spending transaction is settled, the balance of the resulting  fund $F''$ is equal to $F.\fbl/s_2$ which is also the spending amount (which we do not explicitly represent in the protocols). 
\end{proof}

\begin{lemma}
Settlement amounts for payments from $F$ are subtracted from settlement for $F$:  If executing transaction $(F,\mbox{SETTLE})$ results in a fully certified fund $F_R$:

   \[  F_R.\fbl \leq  F.\fbl - \sum_{F' \in \settled_F} F'.\fbl
    \] 
\end{lemma}
\begin{proof}(sketch)
Consider $F' \in  \settled_F$ that results from settling a fund $F_s$ identified by  
transaction $(\tid,N_s)$ when $F_s$ is settled, resulting in $F'$, either $F$ is not yet settled and $(\tid,\textproc{H}(N_s))$ gets added to $\mathit{transactions}[F]$ by the validators of the settlement transaction   or  $(\tid,\textproc{H}(N_s))$ is already in $\mathit{transactions}[F]$. In either case, when $F$ is settled, all these transactions will be in $\mathit{transactions}[F]$ and their amounts deducted from the balance of $F$.
\end{proof}

\begin{lemma}[No more than $s_2 = k_2+f/v_s$ fractional spending transactions]
With high probability, no more than $s_2 = k_2+f/v_s$ fractional spending transaction can be validated, where $v_s$ is the validation slack.
\end{lemma}
\begin{proof}
With high probability, $k_2$ is an upper bound on the number of fractional spending transactions that can be validated if the adversary can only randomly chose the validators to corrupt, which is the case when the seller is honest. This is also the bound if no vaildators are corrupted.
So, we assume that $k_2$ fractional spending transactions are executed with no validator corruption. Any additional transactions would have quorums of validators that intersect the previously selected quorums in $\beta m$ validators excluding previously corrupted validators. It follows that the adversary would need to corrupt $(\beta-\alpha)m$ new validators for each additional fractional spending transactions so that the number of validators that validate the transaction is raised to $\alpha m$. A total of $f/((\beta-\alpha)m$ additional spending transactions is possible to execute this way before running our of new validators to corrupt. The total number of fractional spending transactions possible is therefore $s_2 = k_2 +f/((\beta-\alpha)m$ =  $k_2+f/v_s$. 
\end{proof}
    
\begin{lemma}
If the owners of a fund $F$ are honest, settlement amount is no less than the the initial balance of $F$ minus payments made from $F$: 
    If executing transaction $(F,\mbox{SETTLE})$ results in a fully certified fund $F_R$ and $F.\owners$ is honest:
     \[  F_R.\fbl \geq F.\fbl - \sum_{F' \in \funds_F} F'.\fbl
    \]
\end{lemma}
\begin{proof}(sketch)
This follows directly from how the settlement amount is calculated in which properly validated transactions originating from the owners of $F$ are included in $\mathit{transactions}[F]$ and are deducted from the resulting balance.
\end{proof}

\begin{lemma}[Non-interference] If a total of $s \leq s_1 = k_1$ payment transactions are made from  fully certified fund $F$ and no additional payment or settlement transactions are initiated by $F.\owners$, then,  every one of the $s$ transactions whose seller (payee) is honest will be validated independently of the validations of the other payment transactions from $F$.\end{lemma}
\begin{proof}(sketch)
We consider one of these transactions with an honest seller. The seller will select a quorum of validators according to the protocol and gets validations from the buyer in $F.\owners$. Then the seller will contact the validators to get the transaction validated. By the properties of $(k_1,k_2)$-quorum systems, the seller will get replies from $m - (1+\mu) p_f m$ validators, $(1 - \alpha) \times m$ of which have not previously validated another transaction from $F$ and the transaction will be validated.
\end{proof}

\begin{lemma}
Successful settlement for fully certified funds: If the owner of fully certified fund $F$ is honest and executes an $(F,\mbox{SETTLE})$ transaction, the settlement transaction for $F$ will terminate. 
\end{lemma}
\begin{proof}
Since the owner is honest, at most $s_1 = k_1$ fractional spending transactions from $F$ can be executed. When the $(F,\mbox{SETTLE})$ transaction is executed, the validators will have no more than $s_1$ transactions in $\mathit{transaction}[F]$ because every transaction in $\mathit{transaction}[F]$ must be checked for validity with 
\textproc{ValidPayment}($\tid,h_s,\sigma,N$) which checks that sigma is a valid signature by the buyer for $\tid||h_s||N$. Since the only potentially blocking condition in the validator's buyer settlement code is when there are more than $s_1$ transactions in $\mathit{transaction}[F]$, every honest validator that handles the settlement will validate $(F,\mbox{SETTLE})$ and eventually the buyer will get $n-2f$ identical replies from validators. 
\end{proof}


\section{\texorpdfstring{$(k_{1},k_{2})$}--Quorums Properties}\label{sec:k1k2-proofs}

\begin{L1}
An $(m,n)$ uniform balanced quorum system is a $(k_1,k_2)$-quorum system with $\alpha = 1/3$ and $\beta = 2/3$ has $\epsilon = \delta = \negl[n]$ and {\em validation slack} = $m/3$, if $p_f+\alpha_1 < 1/3$, where $\alpha_1 = k_1m/n$ and $p_f = f/n$.
\end{L1}
\begin{proof}
We need to show that:

\begin{enumerate}
    \item Upper bound: $\Pr_{Q , Q_j \leftarrow  \mathcal{Q}\,;\, j \in J_1; \,|J_1| \leq k_1}[|Q \cap ({\mathcal{F}_{pr}}\cup \bigcup_{j \in J_1} Q_j)| > m/3 ] \leq \negl[n]$
    \item Lower bound: $\Pr_{Q , Q_j \leftarrow  \mathcal{Q}\,;\, j \in J_2; \,|J_2| \geq k_2}[|(Q \cap (\bigcup_{j \in J_2} Q_j))-{\mathcal{F}_{pr}}| \leq 2m/3 ] \leq \negl[n]$
\end{enumerate}

For the upper bound, consider the intersection of a quorum $Q$ with the union of up to $k_1$ previously selected quorums. Since $k_1m/n = \alpha_1$, $\alpha_1 n$ is an upper bound on the size of the union of the $k_1$ previously selected quorums because each of them has size $m$. The expected size of the intersection of $Q$ with these $k_1$ previously selected quorums is at most $\alpha_1|Q| = \alpha_1 m$ because the probability that a given randomly chosen server in $Q$ is in the union is at most $\alpha_1 n / n = \alpha_1$. Similarly, the probability that a server in $Q$ is previously corrupted is at most $p_f$ and the expected number of servers in $Q$ that are previously corrupted is at most $p_f|Q| = p_fm$. So, the expected size of the intersection of $Q$ with the union of the previously corrupted servers together with the union of $k_1$ previously selected quorums is at most $(\alpha_1+p_f)m$. Since $\alpha_1+p_f < 1/3$, $(\alpha_1+p_f)m < m/3$. Let $r = (1/3(\alpha_1+p_f))-1$. This value of $r$ is positive and independent of $m$. For this value of $r$, we have $(1+r)\times (\alpha_1+p_f)m = m/3$. So, the probability that the size of the intersection exceeds by a factor $1+r$ the expected size of the intersection is the same as the probability that the size of the intersection exceeds $m/3$

By Chernoff's upper tail bounds~\cite{probability-book}, the probability that the size of the intersection exceeds by a factor of $1+r$, $r > 0$, the upper bound $(\alpha_1+p_f)m$ on the expected size is:

\[
\Pr_{Q \leftarrow  \mathcal{Q}}[|Q \cap ({\mathcal{F}_{pr}}\cup \bigcup_{j \in J_1} Q_j)| \geq (1+r)\times (\alpha_1+p_f)m = m/3] \leq e^{-r^2 \times (\alpha_1+p_f)m/3 } 
\]
which is a negligible function of $m$. Since $r$ does not depend on $m$, but on the value of $\alpha_1+p_f$   which is a constant for given $p_f$, $k_1$ and $k_2$, the function above decreases exponentially with $m$ for fixed $\alpha_1$ and $p_f$ satisfying $\alpha_1+p_f < 1/3$. Also, for a fixed $k_1$ and $\alpha_1$, $m$ grows linearly with $n$: $m = \alpha_1 n/k_1$. In other words, the probability is a negligible function of $n$ for fixed $k_1$, $\alpha_1$ and $p_f$.

Now, it remains to show that the lower bound requirement holds: \[\Pr_{Q , Q_j \leftarrow  \mathcal{Q}\,;\, j \in J; \,|J| \geq k_2}[|(Q \cap (\bigcup_{j \in J} Q_j))-{\mathcal{F}_{pr}}| < 2m/3 ] \leq  \negl[n]
\]

The expected number of servers in a quorum $Q$ that are in the union of $k_2$ previously selected quorums is $(k_2m/n)|Q| = (1-\alpha_1)m$ because $n = (k_1+k_2)m$ for a uniform system. The expected number of servers in $Q$ that are previously corrupted is less than $p_f m$. So, the expected number of non-corrupted servers in $Q$ that are in one of the last $k_2$ chosen quorums is at least $(1-\alpha_1)m - p_fm = (1 -(\alpha_1+p_f))m > 2m/3$ because $\alpha_1+p_f < 1/3$ as we have noted above. 
If we define $r' = 1 - ((1 -(\alpha_1+p_f))/(2/3)$, we have $0 < r' < 1$, and $(1- r') \times (1 -(\alpha_1+p_f))m = 2m/3$.
Using Chernoff's lower tail bound, we have 
\[\Pr_{Q , Q_j \leftarrow  \mathcal{Q}\,;\, j \in J_2; \,|J_2| \geq k_2}[|(Q \cap (\bigcup_{j \in J} Q_j))-{\mathcal{F}_{pr}}| \leq (1- r) \times (1 -(\alpha_1+p_f))m = 2m/3] \leq e^{-r'^2 (1 -(\alpha_1+p_f))m/2}\]
which is a negligible function of $m$ and, therefore, as we argued above, is also a negligible function of $n$. 
\end{proof}

Even though the proof shows that the probability of not satisfying the $(k_1,k_2)$-quorum requirements is negligible, values of $\alpha_1$ and $p_1$ for which the sum $\alpha_1+p_f$ is very close to 1/3 result in probability bounds that are not useful in practice. For the parameter choices we made, the {\em validation slack} of the system is $m/3$. Different values for $\alpha_1$ and $p_f$ could result in lower validation slack but better overall performance. Determining the optimal combination of parameters is subject of future work.
\begin{L2}
An $(m,n)$ uniform balanced quorum system is a $(k_1,k_2)$-quorum system is a $(k_1,k_2, \epsilon, \delta, \alpha = 1/3, \beta = 2/3, \mu )$-asynchronous quorum system with $\epsilon = \delta = \negl[n]$ and {\em validation slack} = $m/3$, if $p_f = f/n < 1/8$ and $\alpha_1 = k_1m/n < 1/24$.
\end{L2}
\begin{proof}
The proof is almost identical to that of Lemma~\ref{lem:uniform-prob}. Recall that for asynchronous $(k_1,k_2)$-quorum systems, the  intersection properties should hold even if we exclude from a quorum $Q$ a set $Q_s$ of size at most $(1+\mu)p_fm$.

For the upper bound, not including some replies can only make the intersection smaller, so the proof that with negligible probability, the intersection's size is less than $m/3$ carries over to this setting. For the lower bound, the only difference is that we are excluding a set $Q_s$ of size at $(1+\mu)p_f n$ servers (in addition to the already excluded $p_f n$ servers). So, the expected size of non-corrupt servers in $Q-Q_s$ that are also in one of the previous $k_2$ chosen quorums is at least $(1-\alpha_1 - p_f)m - (1+\mu)p_f m = (1-\alpha_1 - (2+ \mu) p_f)m$. We show that with negligible probability the expected size is short by a factor of more than $(1-r)$ of the expected size, where  
 $r = 1 - ((1 -(\alpha_1+p_f+(1+\mu)p_f))/(2/3)$. Similarly to what we did in Lemma~\ref{lem:uniform-prob} for the synchronous case, we have $0 < r < 1$ and 
using Chernoff's lower tail bound, we have for any $Q_s$ such that  $|Q_s| \leq m p_f (1+\mu)$:
\begin{align}
 &\Pr_{Q , Q_j \leftarrow  \mathcal{Q}\,;\, j \in J_2; \,|J_2| \geq k_2}[|((Q-Q_s) \cap ((\bigcup_{j \in J}  Q_j)-{\mathcal{F}_{pr}})| < 2m/3] \\
& =  \Pr_{Q , Q_j \leftarrow  \mathcal{Q}\,;\, j \in J_2; \,|J_2| \geq k_2}[|(Q \cap ((\bigcup_{j \in J} Q_j)-{\mathcal{F}_{pr}-Q_s)}|] < \leq (1- r) \times (1 -(\alpha_1+p_f)m] \\
&  \leq e^{-r^2/2 \times (1 -(\alpha_1+p_f)m}
\end{align}
\end{proof}

\end{document}